\documentclass{sig-alternate}
\usepackage[T1]{fontenc}
\usepackage{microtype}
\usepackage{caption}
\usepackage{hyperref}
\usepackage{ezmath,ezcode}
\usepackage{color,xspace}

\newtheorem{theorem}{Theorem}

\newcommand{\limp}{\rightarrow}
\newcommand{\ltrue}{\mathit{true}}

\newcommand{\ctl}{\textsc{CTL}\xspace}
\newcommand{\ctlfo}{\textsc{CTL+FO}\xspace}
\newcommand{\pathE}{E} 
\newcommand{\pathA}{A} 
\newcommand{\ltlNext}{X} 
\newcommand{\ltlG}{G} 
\newcommand{\ltlF}{F} 
\newcommand{\ltlU}{U}  

\newcommand{\vars}{v}

\newcommand{\init}{\mathit{init}}
\newcommand{\nextRel}{\mathit{next}}
\newcommand{\inv}{\mathit{inv}}
\newcommand{\aux}{\mathit{aux}}
\newcommand{\wfPred}{\mathit{wf}}
\newcommand{\rank}{\mathit{rank}}

\newcommand{\algConsGen}{\textsc{Gen}\xspace}
\newcommand{\ezcase}[1]{$\;\mid$ #1 $\Rightarrow$}

\newcommand{\ehsfTool}{\textsc{Ehsf}\xspace}
\newcommand{\skolemRelSymbol}{\mathit{sk}}

\newcommand{\funTemplateOf}[2]{\textsc{Templ}(#1)(#2)}

\newcommand{\computations}[2]{\ensuremath{\Pi_{#1}(#2)}}
\newcommand{\program}{\ensuremath{P}}
\newcommand{\sat}[3]{\ensuremath{#1, #2 \models #3}}
\newcommand{\subst}[2]{\ensuremath{[#2/#1]}}

\newcommand{\modelsT}{\models_{\mathcal{T}}}

\newcommand{\ourTool}{{\sc CTLFO}\xspace} 
\newcommand{\ctlfosat}{\models_{\mathit{CTL+FO}}}

\newcommand{\cc}[1]{\multicolumn{2}{|c|}{#1}}

\newcommand{\ipYes}{\checkmark\xspace}
\newcommand{\ipNo}{$\times$}
\newcommand{\ipTO}{T/O\xspace}

\usepackage{flushend}

\usepackage{ifthen}
\usepackage{substr}
\newif\ifisReport
\makeatletter
\@ifundefined{paperVersion}{\newcommand{\paperVersion}{paper.pdf}}{}
\IfSubStringInString{report}{\paperVersion}{
  \isReporttrue
}{
  \isReportfalse
}

\begin{document}

\title{CTL+FO Verification as Constraint Solving}

\conferenceinfo{SPIN}{'14, July 21–23, 2014, San Jose, CA, USA}
\CopyrightYear{2014}
\crdata{978-1-4503-2452-6/14/07}

\numberofauthors{3}
\author{
\alignauthor
Tewodros A. Beyene\\
  \affaddr{Technische Universit\"at M\"unchen, Germany}
\alignauthor
Marc Brockschmidt\\
  \affaddr{Microsoft Research Cambridge, UK}
\alignauthor
Andrey Rybalchenko\\
  \affaddr{Microsoft Research Cambridge, UK}
}

\maketitle
\begin{abstract}
Expressing program correctness often requires relating program data
throughout (different branches of) an execution.
Such properties can be represented using \ctlfo, a logic that
allows mixing temporal and first-order quantification.
Verifying that a program satisfies a \ctlfo property is a challenging
problem that requires both temporal and data reasoning.
Temporal quantifiers require discovery of invariants and ranking
functions, while first-order quantifiers demand instantiation
techniques.
In this paper, we present a constraint-based method for proving \ctlfo
properties automatically.
Our method makes the interplay between the temporal and first-order
quantification explicit in a constraint encoding that combines
recursion and existential quantification. 
By integrating this constraint encoding with an off-the-shelf solver
we obtain an automatic verifier for~\ctlfo.


\end{abstract}

\category{D.2.4}{Software Engineering}{Software/Program Verification}
\category{F.3.1}{Logics and Meanings of Programs}{Specifying and Verifying and Reasoning about Programs}
\terms Verification, Theory
\keywords Model Checking, Software Verification, Temporal Logic

\section{Introduction}

In specifying the correct behaviour of systems, relating data at various
stages of a computation is often crucial.
Examples include program termination~\cite{TerminatorPLDI06} (where the value
of a rank function should be decreasing over time), 
correctness of reactive systems~\cite{Hodkinson02} (where each incoming
request should be handled in a certain timeframe),
and information flow~\cite{Heusser10} (where for all possible secret input
values, the output should be the same).
The logic \ctlfo offers a natural specification mechanism for such properties,
allowing to freely mix temporal and first-order quantification.
First-order quantification makes it possible to specify variables dependent on
the current system state, and temporal quantifiers allow to relate this data
to system states reached at a later point.

While \ctlfo and similar logics have been identified as a specification
language before, no fully automatic method to check \ctlfo properties on
infinite-state systems was developed.
Hence, the current state of the art is to either produce verification tools
specific to small subclasses of properties, or using error-prone program
modifications that explicitly introduce and initialize ghost variables, which
are then used in (standard) \ctl specifications.

In this paper, we present a fully automatic procedure to transform a \ctlfo
verification problem into a system of existentially quantified recursive Horn
clauses.
Such systems can be solved by leveraging recent advances in constraint 
solving~\cite{ehsf}, allowing to blend first-order and temporal
reasoning.
Our method benefits from the simplicity of the proposed proof rule and
the ability to leverage on-going advances in Horn constraint solving.

\paragraph{Related Work}
Verification of \ctlfo and its decidability and complexity have been studied
(under various names) in the past.
Bohn et al.~\cite{Bohn98} presented the first model-checking algorithm.
Predicates partitioning a possibly infinite state space are deduced
syntactically from the checked property, and represented symbolically by
propositional variables.
This allows to leverage the efficiency of standard BDD-based model checking
techniques, but the algorithm fails when the needed partition of the state
space is not syntactically derivable from the property.

Working on finite-state systems, 
Hall\'e et al.~\cite{Halle07}, Patthak et al.~\cite{Patthak02} and
Rensink~\cite{Rensink06} discuss a number of different techniques for
quantified CTL formulas.
In these works, the finiteness of the data domain is exploited to instantiate
quantified variables, thus reducing the model checking problem for quantified
CTL to standard CTL model checking.

Hodkinson et al.~\cite{Hodkinson02} study the decidability of \ctlfo and some
fragments on infinite state systems.
They show the general undecidability of the problem, but also identify certain
decidable fragments. Most notably, they show that by restricting first order
quantifiers to state formulas and only applying temporal quantifiers to
formulas with at most one free variable, a decidable fragment can be obtained.
Finally, Da Costa et al.~\cite{DaCosta12} study the complexity of checking
properties over propositional Kripke structures, also providing an overview of
related decidability and complexity results.
In temporal epistemic logic, Belardinelli et al.~\cite{Belardinelli12}
show that checking \textsc{FO-CTLK} on a certain subclass of infinite
systems can be reduced to finite systems.
In contrast, our method directly deals with quantification over
infinite domains.


\section{Preliminaries}

\vspace{-4mm}
\paragraph{Programs}

We model programs as transition systems.
A program $\program$ consists of a tuple of program variables~$\vars$,
an initial condition $\init(\vars)$, and a transition
relation~$\nextRel(\vars, \vars')$. 
A state is a valuation of $\vars$. 
A computation $\pi$ is a maximal sequence of states $s_1, s_2, \ldots$
such that $\init(s_1)$ and for each pair of consecutive states $(s,
s')$ we have $\nextRel(s, s')$.
The set of computations of $\program$ starting in $s$ is denoted
by~$\computations{\program}{s}$.

\vspace{-1mm}
\paragraph{\ctlfo Syntax and Semantics}
The following definitions are standard, see
e.g.~\cite{Bohn98,KestenTCS95}.

Let $\mathcal{T}$ be a first order theory and $\modelsT$ denote its
satisfaction relation that we use to describe sets and relations
over program states.
Let $c$ range over assertions in $\mathcal{T}$ and $x$ range over variables.
A \ctlfo formula $\varphi$ is defined by the following grammar using
the notion of a path formula~$\phi$.
\begin{equation*}
  \begin{array}[t]{@{}r@{\;::=\;}l@{}}
  \varphi &
    \forall x: \varphi \mid
    \exists x: \varphi \mid
    c \mid
    \varphi \land \varphi  \mid
    \varphi \lor \varphi   \mid
    \pathA \, \phi \mid
    \pathE \, \phi \\ [\jot]
    \phi & \ltlNext \varphi \mid \ltlG \varphi \mid \varphi \ltlU \varphi
  \end{array}
\end{equation*}
As usual, we define  $\ltlF \varphi = (\ltrue \ltlU \varphi)$.
The satisfaction relation $\program\models \varphi$ holds if and only
if for each $s$ such that $\init(s)$ we
have~$\sat{\program}{s}{\varphi}$.
We define $\sat{\program}{s}{\varphi}$ as follows using an auxiliary
satisfaction relation $\sat{\program}{\pi}{\phi}$.
Note that $d$ ranges over values from the corresponding domain.
\begin{equation*}
  \begin{array}[t]{@{}l@{\text{ iff }}l@{}}
    \sat{\program}{s}{\forall x: \varphi} 
    &
    \text{for all $d$ holds } \sat{\program}{s}{\varphi\subst{x}{d}} \\[\jot]
    \sat{\program}{s}{\exists x: \varphi} 
    &
    \text{exists $d$ such that }
    \sat{\program}{s}{\varphi\subst{x}{d}} \\[\jot]
    \sat{\program}{s}{c} 
    &
    s \modelsT c \\[\jot]
    \sat{\program}{s}{\varphi_1 \land \varphi_2} 
    &
    \sat{\program}{s}{\varphi_1} \text{ and }
    \sat{\program}{s}{\varphi_2}\\[\jot]
    \sat{\program}{s}{\varphi_1 \lor \varphi_2} 
    &
    \sat{\program}{s}{\varphi_1} \text{ or }
    \sat{\program}{s}{\varphi_2}\\[\jot]
    \sat{\program}{s}{\pathA\, \phi} 
    &
    \text{for all $\pi \in \computations{\program}{s}$ holds } 
    \sat{\program}{\pi}{\phi} \\[\jot]
    \sat{\program}{s}{\pathE\, \phi}
    &
    \text{exists $\pi \in \computations{\program}{s}$ such that }
    \sat{\program}{\pi}{\phi} \\[\jot]
    \sat{\program}{\pi}{\ltlNext \varphi}
    &
    \pi = s_1, s_2, \ldots \text{ and }
    \sat{\program}{s_2}{\varphi} \\[\jot]
    \sat{\program}{\pi}{\ltlG \varphi}
    &
    \pi = s_1, s_2, \ldots \text{for all $i\geq 1$ holds }
    \sat{\program}{s_i}{\varphi} \\[\jot]
    \sat{\program}{\pi}{\varphi_1 \ltlU \varphi_2}
    &
    \begin{array}[t]{@{}l@{}}
      \pi = s_1, s_2, \ldots 
      \text{ and exists $j\geq 1$ such that }\\[\jot]
      \sat{\program}{s_j}{\varphi_2} \text{ and }
      \sat{\program}{s_i}{\varphi_1} \text{ for } 1 \leq i \leq j
    \end{array}
  \end{array}
\end{equation*}

\vspace{-3mm}
\paragraph{Quantified Horn Constraints}

Our method uses the \ehsfTool~\cite{ehsf} solver for forall-exists
Horn constraints and well-foundedness. 
We omit the syntax and semantics of constraints solved by \ehsfTool,
see \cite{ehsf} for details. 
Instead, we consider an example:
\begin{equation*}
  x \geq 0 \limp \exists y: x \geq y \land \mathit{rank}(x, y), \qquad
  \wfPred(\rank).
\end{equation*}
These constraints are an assertion over the interpretation
of the ``query symbol'' $\mathit{rank}$ (the predicate $\wfPred$ is not a
query symbol, but requires well-foundedness).
A solution maps query symbols into constraints.
The example has a solution that maps $\rank(x, y)$
to the constraint $(x \geq 0 \land y \leq x-1)$.

\ehsfTool resolves clauses like the above using  a CEGAR 
scheme to discover witnesses for existentially quantified variables. 
The refinement loop collects a global constraint that declaratively
determines which witnesses can be chosen. The chosen witnesses are
used to replace existential quantification, and then the resulting
universally quantified clauses are passed to a solver over decidable
theories, e.g., \textsc{HSF}~\cite{GrebenshchikovPLDI12} or
$\mu$\textsc{Z}~\cite{muz}. Such a solver either finds a solution, i.e., a
model for uninterpreted relations constrained by the clauses, or
returns a counterexample, which is a resolution tree (or DAG)
representing a contradiction. \ehsfTool turns the counterexample into
an additional constraint on the set of witness candidates, and
continues with the next iteration of the refinement loop. 

For the existential clause above, \ehsfTool introduces a
witness/Skolem relation $\skolemRelSymbol$ over variables $x$ and
$y$, i.e., $x\geq0 \land \skolemRelSymbol(x,y) \limp x \geq y \land
\mathit{rank}(x, y)$.
In addition, since for each $x$ such that $x\geq0$ holds we need a
value $y$, we require that such $x$ is in the domain of the Skolem
relation using an additional clause $x\geq0 \limp \exists y:
\skolemRelSymbol(x,y)$.
In the \ehsfTool approach, the search space of a Skolem relation 
$\skolemRelSymbol(x,y)$ is restricted by a template function
$\funTemplateOf{\skolemRelSymbol}{x,y}$.
To conclude this example, we note that one possible solution returned
by \ehsfTool is the Skolem relation $\skolemRelSymbol(x,y) = (y\leq x-1)$.


\section{Constraint Generation}
\label{sec-consgen}

In this section we present our algorithm \algConsGen for generating
constraints that characterize the satisfaction of a \ctlfo formula.
We also consider its complexity and correctness and present an example.

See Figure~\ref{fig-consgen-fo}.
\algConsGen performs a top-down, recursive descent through the syntax
tree of the given \ctlfo formula.
It introduces auxiliary predicates and generates a sequence of
implication and well-foundedness constraints over these predicates.
We use ``,'' to represent the concatenation operator on sequences of
constraints.
At each level of recursion, \algConsGen takes as input a \ctlfo
formula~$\varphi_0$, a tuple of variables $v_0$ that are considered to
be in scope and define a state, assertions $\init(v_0)$ and
$\nextRel(v_0, v_0')$ that describe a set of states and a transition
relation, respectively.
We assume that variables bound by first-order quantifiers in $\varphi_0$ do
not shadow other variables.
To generate constraints for checking if $\program = (\vars,
\init(\vars), \nextRel(\vars, \vars'))$ satisfies $\varphi$ we
execute~$\algConsGen(\varphi, \vars, \init(\vars), \nextRel(\vars,
\vars'))$.

\begin{figure}[t]
\begin{ezcode}
$\algConsGen(\varphi_0, \vars_0, \init(\vars_0), \nextRel(\vars_0, \vars_0')) =$ \[
match $\varphi_0$ with 
\ezcase{$\forall x : \varphi_1$}  \[
let $v_1  = (v_0, x)$ in
$\algConsGen(\varphi_1, \vars_1, \init(\vars_0), \nextRel(\vars_0, \vars_0')\land x' = x)$
\]
\ezcase{$\exists x : \varphi_1$}\[
let $v_1  = (v_0, x)$ in
 let $\aux = \text{fresh symbol of arity } |\vars_1|$ in
$\init(\vars_0) \limp \exists x: \aux(\vars_1), $ 
$\algConsGen(\varphi_1, \vars_1, \aux(\vars_1), \nextRel(\vars_0, \vars_0') \land x'=x)$
\]
\ezcase{$c$} \[
$\init(\vars_0) \limp c$
\]
\ezcase{$\pathE\ltlF \varphi_1$} \[
let $\inv, \aux = \text{fresh symbols of arity } |\vars_0|$ in
let $\rank = \text{fresh symbol of arity } |\vars_0|+|\vars_0|$ in  
$\init(\vars_0) \limp \inv(\vars_0)$,  
$\inv(\vars_0) \land \neg \aux(\vars_0) \limp \exists \vars_0': \nextRel(\vars_0, \vars_0') \land \inv(\vars_0')$
\hspace{3.7cm}$\land\ \rank(\vars_0, \vars_0'),$
$\wfPred(\rank), $
$\algConsGen(\varphi_1, \vars_0, \aux(\vars_0), \nextRel(\vars_0, \vars_0'))$
\]
\]
\end{ezcode}
  \vspace{-2mm}
  \caption{Constraint generation rules for FO quantification, 
    assertions, and existential/eventually quantification.}
  \label{fig-consgen-fo}
 \vspace{-4mm}
\end{figure}


\paragraph{Handling First-Order Quantification}
When $\varphi_0$ is obtained from some $\varphi_1$ by universally
quantifying over $x$, we directly descend into $\varphi_1$ after
adding $x$ to the scope. 
Hence, the recursive call to \algConsGen uses $v_1 = (v_0, x)$.
Since $\init(v_0)$ defines a set of states over $v_1$ in which $x$
ranges over arbitrary values, 
the application $\algConsGen(\varphi_1, v_1, \init(v_0), \dots)$
implicitly requires that $\varphi_1$ holds for arbitrary~$x$.
Since the value of $x$ is arbitrary but fixed within $\varphi_1$, we
require that the transition relation considered by the recursive calls
does not modify $x$ and thus extend $\nextRel$ to $\nextRel(v_0, v_0') \land x'=x$
in the last argument.

When $\varphi_0$ is obtained from some $\varphi_1$ by existentially
quantifying over $x$, we use an auxiliary predicate $\aux$ that implicitly
serves as witness for $x$.
A first constraint connects the set of states $\init(v_0)$ on which
$\varphi_0$ needs to hold with $\aux(v_1)$, which describes the states on
which $\varphi_1$ needs to hold. We require that for every state $s$ allowed
by $\init(v_0)$, a choice of $x$ exists such that the extension of $s$ with
$x$ is allowed by $\aux(v_1)$.
Then, the recursive call $\algConsGen(\varphi_1, v_1, \aux(v_1),
\dots)$ generates constraints that keep track of satisfaction of
$\varphi_1$ on arbitrary $x$ allowed by~$\aux(v_1)$.
Thus, $\aux(v_1)$ serves as a restriction of the choices allowed for $x$.

\paragraph{Handling Temporal Quantification}
We use a deductive proof system for \ctl~\cite{KestenTCS95} and
consider its proof rules from the perspective of constraint generation.

When $\varphi_0$ is a background theory assertion, i.e., does not
use path quantification, \algConsGen produces a constraint
that requires $\varphi_0$ to hold on every initial state.

When $\varphi_0$ requires that there is a path on which $\varphi_1$ eventually
holds, then \algConsGen uses an auxiliary predicate $\aux(\vars_0)$ to
describe those states in which $\varphi_1$ holds.
\algConsGen applies a combination of inductive reasoning together with
well-foundedness to show that $\aux(\vars_0)$ is eventually reached from the
initial states.
The induction hypothesis is represented as $\inv(\vars_0)$ and is required to
hold for every initial state and when $\aux(\vars_0)$ is not reached yet.
Then, the well-foundedness condition $\wfPred$, which requires that it is not
possible to come back into the induction hypothesis forever, ensures that
eventually we reach a ``base case'' in which $\aux(\vars_0)$ holds.
Hence, eventually $\varphi_1$ holds on some computations.
Note that the induction hypothesis $\inv(\vars_0)$, the well-founded relation
$\rank(\vars_0, \vars'_0)$, and the predicate $\aux(\vars_0)$
are left for the solver to be discovered.

\ifisReport
See Appendix~\ref{sec-full-ctl-fo-consgen} for the remaining rules that
describe the full set of \ctl temporal quantifiers.
\else
See \cite{Report} for the remaining rules that describe the full set of \ctl
temporal quantifiers.
\fi

\paragraph{Complexity and Correctness}

\algConsGen performs a single top-down descent through the syntax tree
of the given \ctlfo formula $\varphi$.
The run time and the size of the generated constraints
is linear in the size of $\varphi$.
Finding a solution for the generated constraints is undecidable in
general.
In practice however, the used solver often succeeds in finding a solution
(cf. Sect. \ref{sect:Eval}).
We formalize the correctness of \algConsGen in the following theorem.
\begin{theorem}
For a given program $\program$ with $\init(v)$ and $\nextRel(v, v')$ over $v$ and
a \ctlfo formula $\varphi$ the application $\algConsGen(\varphi, v,
\init(v), \nextRel(v, v'))$ computes a constraint that is satisfiable if
and only if $\program \models \varphi$.
\end{theorem}
\begin{proof}(sketch)
We omit the full proof here for space reasons.
We proceed by structural induction, as the constraint generation of the
algorithm $\algConsGen$.
Formally, we prove that the constraints generated by 
$\algConsGen(\varphi_0, v_0, \init(v_0), \nextRel(v_0, v_0'))$ have a solution 
if and only if the program $P = (v_0, \init(v_0), \nextRel(v_0, v_0'))$ satisfies 
$\varphi_0$.
The base case, i.e., $\varphi_0$ is an assertion $c$ from our background theory
$\mathcal{T}$, is trivial.

As example for an induction step, we consider
$\varphi_0 = \exists x : \varphi_1$.
To prove soundness, we assume that the generated constraints have a solution.
For the predicate $\aux$, this solution is a relation $S_{\aux}$ that satisfies
all constraints generated for $\aux$.
For each $s$ with $\init(s)$, we choose $\overline{x}_s$ such that 
$(s, \overline{x}_s) \in S_{\aux}$.
As we require $\init(v_0) \to \exists x : \aux(v_0, x)$, this element is
well-defined.
We now apply the induction hypothesis for 
$P' = ((v_0, x), \aux(v_0, x), \nextRel(v_0, v_0') \land x' = x)$
and $\varphi_1$. 
Then for all $s$ with $\init(s)$, we have $\sat{P'}{(s,
\overline{x}_s)}{\varphi_1}$, and as $P'$ is not changing $x$ by construction,
also $\sat{P'}{(s, \overline{x}_s)}{\varphi_1\subst{x}{\overline{x}_s}}$.
From this, $\sat{P}{s}{\varphi_0}$ directly follows.

For completeness, we proceed analogously. 
If $\sat{P}{\varphi_0}$ holds, then a suitable instantiation $\overline{x}_s$ 
of $x$ can be chosen for each $s$ with $\init(s)$, and thus we can construct a
solution for $\aux(v_0, x)$ from $\init(v_0)$.
\qed
\end{proof}

\vspace{-3mm}
\paragraph{Example}
\label{sec-example}

We illustrate \algConsGen (see Figure~\ref{fig-consgen-fo}) on a simple
example.
We consider a property that the value stored in a register $v$ can
grow without bound on some computation.
\begin{equation*}
  \forall x: v = x \limp \pathE \ltlF (v > x)
\end{equation*}
This property can be useful for providing evidence that a program is
actually vulnerable to a denial of service attack. 
Let $\init(\vars)$ and $\nextRel(\vars, \vars')$ describe a program over
a single variable~$v$. 

We apply \algConsGen on the property and the program and
obtain the following application trace (here, we treat $\limp$ as expected,
as its left-hand side is a background theory atom).
\begin{equation*}
  \small
  \begin{array}[t]{@{}l@{}l@{}l@{}l@{}l@{}}
    \algConsGen(& \forall x\!:\!v\!=\!x \limp \pathE \ltlF (v\!>\!x), &\vars, & \init(\vars), & \nextRel(\vars, \vars')) \\[\jot]
    \algConsGen(& v\!=\!x \limp \pathE \ltlF ( v\!>\!x), &(\vars, x), & \init(\vars), & \nextRel(\vars, \vars') \land x'\!=\!x) \\[\jot]
    \algConsGen(& v\!=\!x \limp \aux(v, x), &(\vars, x), & \init(\vars), & \nextRel(\vars, \vars') \land x'\!=\!x) \\[\jot]
    \algConsGen(& \pathE \ltlF (v\!>\!x), &(\vars, x),& \aux(\vars, x), &\nextRel(\vars, \vars') \land x'\!=\!x)
  \end{array}
\end{equation*}
This trace yields the following constraints.
\begin{equation*}
  \begin{array}[t]{@{}l@{}}
    \init(\vars) \limp ( v = x \limp \aux(\vars)) \\[\jot]
    \aux(\vars) \limp \inv(\vars, x) \\[\jot]
    \inv(\vars, x) \land \neg (v\!>\!x) \limp 
    \exists \vars', x'\!:\!
    \begin{array}[t]{@{}l@{}}
        \nextRel(\vars, x, \vars', x') \land x'=x\\[\jot]
        \mathrel{\land} \inv(\vars'\!, x') \land \rank(\vars, x, \vars'\!, x') \\ [\jot]
    \end{array}\\
    \wfPred(\rank)
  \end{array}
\end{equation*}
Note that there exists an interpretation of $\aux$, $\inv$, and
$\rank$ that satisfies these constraints if and only if the program
satisfies the property.



\section{Evaluation}
\label{sect:Eval}

In this section we present \ourTool, a \ctlfo verification engine.
\ourTool implements the procedure \algConsGen and 
applies \ehsfTool~\cite{ehsf} to solve resulting clauses.
\pagebreak

\newcommand{\Eric}[1]{#1}
\newcommand{\diff}[1]{#1} 

\begin{table}[!t]
\small
\vspace{-2mm}
\caption{Evaluation of \ourTool on examples
 from~\cite{CookPLDI13}.}
\label{table-systems}
\centering
\begin{tabular}{@{}|@{}l@{}|@{}c@{}|@{}c@{}|@{}c@{}|@{}c@{}|@{}c@{}|}
  \hline
  & Property $\phi$ & \cc{\!\!$\ctlfosat \phi$\!\!} & \cc{\!\!$\ctlfosat\neg\phi$\!\!} \\
  \cline{3-4} \cline{5-6}
  & & Res. & Time & Res. & Time \\
  \hline 
  P1  & $\exists x: AG(a=x \limp AF(r=1))$ & \ipYes  &  1.0& \ipNo   &  0.1 \\  
      & $AG(\exists x: a=x \limp AF(r=1))$ & \ipYes  &  0.9& \ipNo   &  0.1 \\  
  \hline 
  P2  & $\exists x: EF(a=x \land EG(r\neq5))$  & \ipYes  & 0.9 & \ipNo  & 0.2 \\
      & $EF(\exists x: a=x \land EG(r\neq5))$  & \ipYes  & 0.6 & \ipNo  & 0.2 \\
  \hline 
  P3  & $\exists x: AG(a=x \limp EF(r=1))$  & \ipYes  &  1.1& \ipNo  & 0.1 \\
      & $AG(\exists x: a=x \limp EF(r=1))$  & \ipYes  &  1.0& \ipNo  & 0.1 \\
  \hline 
  P4  & $\exists x: EF(a=x \land AG(r\neq1))$ & \ipYes  & 1.8 & \ipNo  & 0.4  \\
      & $EF(\exists x: a=x \land AG(r\neq1))$ & \ipYes  & 0.9 & \ipNo  & 0.4  \\
  \hline 
  \hline 
  P5  & $\exists x: AG(s=x \limp AF(u=x))$  & \ipYes & 7.0& \ipNo & 0.1 \\
      & $AG(\exists x: s=x \limp AF(u=x))$  & \ipYes & 7.2& \ipNo & 0.1 \\
  \hline 
  P6  & $\exists x: EF(s=x \land EG(u\neq x))$ & \ipYes & 1.8 & \ipNo & 2.2 \\
      & $EF(\exists x: s=x \land EG(u\neq x))$ & \ipYes & 1.1 & \ipNo & 2.1 \\
  \hline 
  P7  & $\exists x: AG(s=x \limp EF(u=x))$ & \ipYes & 3.1& \ipNo & 0.2 \\
      & $AG(\exists x: s=x \limp EF(u=x))$ & \ipYes & 6.5& \ipNo & 0.1 \\
  \hline 
  P8  & $\exists x: EF(s=x \land AG(u\neq x))$  & \ipYes & 14.3 & \ipNo & 1.8 \\
      & $EF(\exists x: s=x \land AG(u\neq x))$  & \ipYes & 13.9 & \ipNo & 1.8 \\
  \hline 
  \hline 
  P9  & $\exists x: AG(a=x\limp AF(r=1))$ & \ipYes   & 118.7 & \ipNo  &17.3 \\
      & $AG(\exists x: a=x\limp AF(r=1))$ & \ipYes   & 82.3  & \ipNo  & 1.4 \\
  \hline 
  P10 & $\forall x: EF(a=x\land EG(r\neq1))$ & \ipTO &- & \ipNo  & 3.5 \\
      & $EF(\forall x: a=x\land EG(r\neq1))$ & \ipTO &- & \ipNo  & 3.5 \\
  \hline 
  P11 & $\exists x: AG(a=x\limp EF(r=1))$ & \ipYes  & 126.8 & \ipNo  & 3.6 \\
      & $AG(\exists x: a=x\limp EF(r=1))$ & \ipYes  & 140.3 & \ipNo  & 0.2 \\
  \hline 
  P12 & $\forall x: EF(a=x\land AG(r\neq1))$ & \ipYes& 146.7 & \ipNo & 3.2 \\
      & $EF(\forall x: a=x\land AG(r\neq1))$ & \ipYes& 161.7 & \ipNo & 0.2  \\
  \hline 
  \hline 
  P13 & $\exists x: AF(io = x)\lor AF(ret=x)$        & \ipYes  & 576.8 & \ipNo    &  0.3   \\
  \hline 
  P14 & $\exists x: EG(io\neq x)\land EG(ret\neq x)$ & \ipYes  &  15.1 & \ipNo    & 48.1 \\
  \hline 
  P15 & $\exists x: EF(io=x)\land EF(ret=x)$         & \ipYes  & 166.4 & \ipNo    &  1.9 \\
  \hline 
  P16 & $\exists x: AG(io\neq x)\lor AG(ret\neq x)$  & \ipYes  &   3.4 & \ipTO    &  - \\
  \hline
\end{tabular}
\vspace{-5mm}
\end{table}


We run \ourTool on the examples \texttt{OS frag.1}, \ldots, \texttt{OS frag.4}
from industrial code from \cite[Figure~7]{CookPLDI13}.
Each example consists of a program and a \ctl property to be proven.
We have modified the given properties to lift the \ctl formula to \ctlfo.
As example, consider the property $AG(a=1\rightarrow AF(r=1))$.
One modified property to check could be $\exists x: AG(a=x\rightarrow
AF(r=1))$, and another one is $AG(\exists x: (a=x\rightarrow
AF(r=1)))$.  
By doing similar satisfiability-preserving transformations of the properties
for all the example programs, we get a set programs whose properties are
specified in \ctlfo as shown in Table~\ref{table-systems}.
For each pair of a program and \ctlfo property $\phi$, we generated two
verification tasks: proving $\phi$ and proving $\neg\phi$.
While the existence of a proof for a property $\phi$ implies that $\neg\phi$
is violated by the same program, we consider both properties to show the
correctness of our tool.

We report the results in Table~\ref{table-systems}. \ipYes (resp. \ipNo{})
marks cases where \ourTool was able to prove (resp. disprove) a \ctlfo
property. \ipTO marks the cases where \ourTool could not find a
solution or a counter-example in 600 seconds.
 
\ourTool is able to find proofs for all the correct programs except for
P10 and counter-examples for all incorrect programs except for P16. 
Currently, \ourTool models the control flow symbolically using a program
counter variable, which we believe is the most likely reason for the
solving procedure to time out. 
Efficient treatment of control flow along the lines of explicit
analysis as performed in the CPAchecker framework could lead to
significant improvements for dealing with programs with large
control-flow graphs~\cite{ExplicitFASE13}. 
An executable of \ourTool, together with a more verbose evaluation, can be found at
\url{https://www7.in.tum.de/~beyene/ctlfo/}. 

For cases where the property contains nested path quantifiers
and the outer temporal quantifier is $F$ or $U$, our implementation
may generate non-Horn clauses following the proof system
from~\cite{KestenTCS95}.
While a general algorithm for solving non-Horn clauses is beyond the
scope of this paper, we used a simple heuristic to seed solutions for
queries appearing under the negation operator.


\section{Conclusion}

This paper presented an automated method for proving program
properties written in the temporal logic \ctlfo, which combines
universal and existential quantification over time and data.
Our approach relies on a constraint generation algorithm that follows
the formula structure to produce constraints in the form of Horn
constraints with forall/exists quantifier alternation.
The obtained constraints can be solved using an off-the-shelf
constraint solver, thus resulting in an automatic verifier.

\bibliographystyle{abbrv}
\bibliography{biblio}


\newpage
\appendix

\section{Remaining rules}
\label{sec-full-ctl-fo-consgen}
In this section we present the remaining rules of \algConsGen, which
deal with the complete set of temporal quantifiers. 
See Figure~\ref{fig-consgen-temporal}.\\

\begin{minipage}{\textwidth}
\begin{ezcode}
\ezcase{$\pathA\, \ltlNext \varphi_1$} \[
let $\aux = \text{fresh symbol of arity } |\vars_0|$ in
$\init(\vars_0) \limp \exists \vars_0': \nextRel(\vars_0, \vars_0'), $ 
$\init(\vars_0) \land \nextRel(\vars_0, \vars_0')\limp \aux(\vars_0'), $ 
$\algConsGen(\varphi_1, \vars_0, \aux(\vars_0), \nextRel(\vars_0, \vars_0'))$
\]
\ezcase{$\pathE\, \ltlNext \varphi_1$} \[
let $\aux = \text{fresh symbol of arity } |\vars_0|$ in
$\init(\vars_0) \limp \exists \vars_0': \nextRel(\vars_0, \vars_0')\land \aux(\vars_0'), $ 
$\algConsGen(\varphi_1, \vars_0, \aux(\vars_0), \nextRel(\vars_0, \vars_0'))$
\]
\ezcase{$\pathA\, \ltlG \varphi_1$} \[
let $\inv = \text{fresh symbol of arity } |\vars_0|$ in
$\init(\vars_0) \limp \inv(\vars_0),$ 
$\inv(\vars_0) \land \nextRel(\vars_0, \vars_0') \limp \inv(\vars_0'), $  
$\algConsGen(\varphi_1, \vars_0, \inv(\vars_0), \nextRel(\vars_0, \vars_0'))$
\]
\ezcase{$\pathE\, \ltlG \varphi_1$} \[
let $\inv = \text{fresh symbol of arity } |\vars_0|$ in
$\init(\vars_0) \limp \inv(\vars_0),$
$\inv(\vars_0) \land \nextRel(\vars_0, \vars_0') \limp \exists \vars_0': \nextRel(\vars_0, \vars_0') \land \inv(\vars_0'), $  
$\algConsGen(\varphi_1, \vars_0, \inv(\vars_0), \nextRel(\vars_0, \vars_0'))$
\]
\ezcase{$\pathA\, (\varphi_1\ltlU \varphi_2)$} \[
let $\inv, \aux_1, \aux_2 = \text{fresh symbols of arity } |\vars_0|$ in
let $\rank = \text{fresh symbol of arity } |\vars_0|+|\vars_0|$ in  
$\init(\vars_0) \limp \inv(\vars_0), $
$\inv(\vars_0) \land \neg \aux_2(\vars_0) \limp  %
\aux_1(\vars_0) \land \exists \vars_0': \nextRel(\vars_0, \vars_0'), $
$\inv(\vars_0) \land \neg \aux_2(\vars_0) \land \nextRel(\vars_0, \vars_0') \limp  %
\inv(\vars_0') \land \rank(\vars_0, \vars_0'),$  
$\wfPred(\rank), $ 
$\algConsGen(\varphi_1, \vars_0, \aux_1(\vars_0), \nextRel(\vars_0, \vars_0')), $ %
$\algConsGen(\varphi_2, \vars_0, \aux_2(\vars_0), \nextRel(\vars_0, \vars_0'))$
\]
\ezcase{$\pathE\, (\varphi_1\ltlU \varphi_2)$} \[
let $\inv, \aux_1, \aux_2 = \text{fresh symbols of arity } |\vars_0|$ in
let $\rank = \text{fresh symbol of arity } |\vars_0|+|\vars_0|$ in  
$\init(\vars_0) \limp \inv(\vars_0),  $
$\inv(\vars_0) \land \neg \aux_2(\vars_0) \limp \aux_1(\vars_0) \land \exists \vars_0': \nextRel(\vars_0, \vars_0') \land %
\inv(\vars_0') \land \rank(\vars_0, \vars_0'),$
$\wfPred(\rank), $
$\algConsGen(\varphi_1, \vars_0, \aux_1(\vars_0), \nextRel(\vars_0, \vars_0')),$ %
$\algConsGen(\varphi_2, \vars_0, \aux_2(\vars_0), \nextRel(\vars_0, \vars_0'))$
\]
\ezcase{$(\pathA/\pathE)\, \ltlF \varphi_1$} %
$\algConsGen(\vars_0, \init(\vars_0), \nextRel(\vars_0, \vars_0'), (\pathA/\pathE)\, (\ltrue \ltlU \varphi_1))$ 
\ezcase{$\varphi_1 \land/\lor \varphi_2$} \[
let $\aux_1, \aux_2 = \text{fresh symbols of arity } |\vars_0|$ in
$\init(\vars_0) \limp \aux_1(\vars_0) \land/\lor \aux_2(\vars_0), $
$\algConsGen(\varphi_1, \vars_0, \aux_1(\vars_0), \nextRel(\vars_0, \vars_0')), $ %
$\algConsGen(\varphi_2, \vars_0, \aux_2(\vars_0), \nextRel(\vars_0, \vars_0'))$ 
\]
\end{ezcode}
\captionof{figure}{Remaining rules of constraint generation algorithm \algConsGen.
  }
  \label{fig-consgen-temporal}
\end{minipage}
\ \\
\ \\
\ \\
\ \\
\ \\
\ \\



\end{document}